\documentclass[12pt]{amsart}
\usepackage{amsthm}
\usepackage[margin=1in]{geometry}
\usepackage{bbm}
\usepackage{mathrsfs}
\usepackage{bm}
\usepackage{color}
\usepackage{enumitem}
\usepackage{tikz}
\usepackage{cite}

\numberwithin{equation}{section}

\newtheorem{theorem}{Theorem}[section]
\newtheorem{lemma}[theorem]{Lemma}
\newtheorem{prop}[theorem]{Proposition}
\newtheorem{coro}[theorem]{Corollary}

\newtheorem{claim}[theorem]{Claim}

\theoremstyle{definition}
\newtheorem{definition}[theorem]{Definition}

\newtheorem{remark}[theorem]{Remark}
\newenvironment{claimproof}[1][Proof of Claim]{\noindent \underline{#1.} }{\hfill$\diamondsuit$}
\theoremstyle{plain}

\newcommand{\bbC}{{\mathbbm{C}}}

\newcommand{\bbN}{{\mathbbm{N}}}

\newcommand{\bbR}{{\mathbbm{R}}}
\newcommand{\bbT}{{\mathbbm{T}}}
\newcommand{\bbZ}{{\mathbbm{Z}}}
\newcommand{\Z}{{\mathbbm{Z}}}
\newcommand{\C}{{\mathbbm{C}}}

\newcommand{\scrF}{{\mathscr{F}}}

\newcommand{\scrH}{{\mathscr{H}}}

\newcommand{\set}[1]{\left\{#1\right\}}

\newcommand{\tlambda}{{\widetilde{\lambda}}}

\newcommand{\tf}{\widetilde{f}}

\newcommand{\subscript}[2]{$#1 _ #2$}

\definecolor{purple}{rgb}{.5,0,1}
\definecolor{green}{rgb}{0,.5,0}

\def\cold#1{\textcolor{blue}{#1}}

\begin{document}
\title[Irreducibility of the Bloch Variety]{Irreducibility of the Bloch Variety for Finite-Range Schr\"odinger Operators}
\author[J.\ Fillman]{Jake Fillman}
\email{fillman@txstate.edu}
\address{Department of Mathematics, Texas State University, San Marcos, TX 78666, USA}
\thanks{J.\ F.\ was supported in part by Simons Foundation Collaboration Grant \#711663. W.\ L.\ and R.\ M. were partially  supported by NSF  DMS-2000345 and DMS-2052572. }

\author[W.\ Liu]{Wencai Liu}
\email{liuwencai1226@gmail.com; wencail@tamu.edu}
\address{Department of Mathematics, Texas A\&M University, College Station TX, 77843, USA.}

\author[R.\ Matos]{Rodrigo Matos}
\email{matosrod@tamu.edu}
\address{Department of Mathematics, Texas A\&M University, College Station TX, 77843, USA.}

\maketitle

\begin{abstract}
We study the Bloch variety of discrete Schrödinger operators associated with a complex periodic potential and a general finite-range interaction, showing that the Bloch variety is irreducible for a wide class of lattice geometries in arbitrary dimension. Examples include the triangular lattice and the extended Harper lattice.
\end{abstract}

\section{Introduction}

\subsection{Setting and Main Theorem}
We will study periodic finite-range Schr\"{o}dinger operators of the form
\begin{equation} \label{eq:A+Vdef} H  = A + V, \end{equation}
acting in $\ell^2(\bbZ^d)$, where $V$ is periodic and $A$ is a Toeplitz operator given by
\[ [A\psi]_n = \sum_{m \in \bbZ^d} a_{n-m} \psi_m.\]
Here, $\{a_n\}_{n\in\bbZ^d}$ is finitely supported
and $V$ will as usual denote both the potential $V:\bbZ^d \to \bbC$ and the corresponding multiplication operator $[V\psi]_n = V_n\psi_n$. 
We say that $V$ is $q$-periodic for $q=(q_1,\ldots,q_d) \in \bbN^d$ if $V_{n+q_je_j} = V_n$ for all $n \in \bbZ^d$ and each $1 \le j \le d$, where $e_j$ denotes the standard $j$th basis vector.

In particular, let us note that the approach discussed herein does not rely on reality of the potential or self-adjointness of $A$. The case in which
\[ a_n = \begin{cases} - 1 & n = \pm e_j \text{ for some } 1 \le j \le d \\
0 & \text{otherwise} \end{cases} \]
corresponds to $A = -\Delta$, the discrete Laplacian.

Our main result is irreducibility of the Bloch variety for all operators of the form \eqref{eq:A+Vdef} subject to a suitable condition on $A$. 
In particular, under mild assumptions on $\{a_n\}_{n \in \bbZ^d}$, the result holds universally for all periodic $V$, including complex-valued potentials.
We will define the Bloch variety precisely later in the manuscript (see Section~\ref{sec:definitions}); for now, the reader may think of it as a relation between the energy, $\lambda$, and the quasi-momentum, $k$, which is given by the zero set of a function $\mathcal{P}(z,\lambda)$ that is a polynomial in $\lambda$ and a Laurent polynomial in $z = e^{2\pi i k}$.

Let us describe some of the main objects and assumptions used in this work. Starting with $A$ we generate the Laurent polynomial
\begin{equation}\label{eq:pdefinition}
p(z) = p_A(z) = \sum_{n \in \bbZ^d} a_n z^{-n},
    \end{equation}
where we employ the standard multi-index notation $z^n = z_1^{n_1}\cdots z_d^{n_d}$.

Let us state our assumptions here in a moderately informal manner. For further definitions, details, and a more precise account, we refer the reader to Section~\ref{sec:mainresult}, where we define the component of lowest degree and the fundamental domain $W$. Let $h$ denote the lowest degree component of $p$ in the sense that $p(z)=h(z)+\text{higher order terms}$. Given  $q = (q_1,\ldots,q_d) \in \bbN^d$ and $n = (n_1,\ldots,n_d) \in \bbZ^d$, the vector $\mu_n=(\mu_n^1,\ldots,\mu_n^d)$ is defined by
\[\mu_n^j = e^{2\pi i n_j/q_j}, \quad 1 \le j \le d.\]
Given $z \in \bbC^d$, we define $\mu_n \odot z = (\mu_n^1z_1,\ldots,\mu_n^dz_d)$.
Our main assumptions are the following:
\begin{enumerate}
\item[($A_1$)] The degree of $h$ is negative.
\item[($A_2$)] The polynomials $h(\mu_n\odot z)$, $n \in W$, are pairwise distinct (cf.~\eqref{eq:fundcelldef}, \eqref{action}, and \eqref{eq:characterActionDef}).
\end{enumerate}

\begin{theorem}\label{t:blochIrr}
	Let $q=(q_1,q_2,\ldots,q_d)$ be given and let $V$ be $q$-periodic.
If $p_A$ satisfies Assumptions~{\rm\ref{assump1}} and {\rm\ref{assump2}}, then the Bloch variety of $H = A+V$ is irreducible modulo periodicity.
\end{theorem}

\begin{remark}
Let us make a few comments about Theorem~\ref{t:blochIrr}.
\begin{enumerate}
\item  The precise definitions of the Bloch variety and irreducibility modulo periodicity are given in Section~\ref{sec:definitions}. To prove Theorem \ref{t:blochIrr}, we use standard Floquet theory to obtain a Laurent polynomial  $\mathcal{P}(z,\lambda)$ with the property that $(k,\lambda)$ belongs to the Bloch variety if and only if $\mathcal{P}(e^{2\pi i k},\lambda) = 0$. We show that $\mathcal{P}$ is irreducible as a function of $z$ and $\lambda$. Since the Bloch variety is defined as a set of pairs $(k,\lambda)$, it is only irreducible after one quotients out the relevant $\Z^d$ action. This is indeed necessary as shown by the free Laplacian, see Equation~(1.22) and Figure~4 in \cite{Kuchment2016BAMS}. The main idea of the proof is to reduce from the operator $A+V$ to $A$ by focusing on the lowest-degree component of $\mathcal{P}$ (after a suitable change of variables). We show that reducibility combined with Assumptions~\ref{assump1} and \ref{assump2} would imply mutually contradictory properties of the lowest-degree component.
\medskip

\item Assumption~\ref{assump1} only depends on $p_A$, whereas Assumption \ref{assump2} depends on $p_A$ and $q$ (via the action of the vectors $\mu_n$). As the reader can see, our proof does not require $V$ to be real-valued.
\medskip 

\item  The strength of the result comes from the generality of the operators under consideration. For instance, this can handle operators in higher dimensions on rather general graphs; specifically, one can handle $\Z^d$-periodic graphs for which $\Z^d$ acts transitively on vertices (see Remark~\ref{rem:zdgraph} for further details). We emphasize that our result does not necessarily imply irreducibility of the Fermi varieties associated with $A+V$.
\end{enumerate}
\end{remark}

 Theorem~\ref{t:blochIrr} is the main motivation for this work. It will follow from a more general result formulated in Theorem~\ref{mainthm} below. 
 
The above assumptions are satisfied and straightforward to verify in many cases of interest. To illustrate scope of applications, we enumerate some corollaries.

We first note that Theorem \ref{t:blochIrr}   provides a direct proof of the irreduciblity of the Bloch variety  for all discrete Schr\"odinger operators on $\bbZ^d$.
\begin{coro} \label{coro:square}
If $A = -\Delta$ denotes the Laplacian on $\ell^2(\Z^d)$, then for any periodic $V$, the Bloch variety of $A+V$ is irreducible modulo periodicity.
\end{coro}
The corollary above was already known via results about the Fermi variety -- see the discussion in Section~\ref{sec:definitions} for additional details and references. Thus, we supply an alternative argument, working directly on the Bloch variety.
 
More significantly, Theorem \ref{t:blochIrr} also enables one to prove irreducibility of the Bloch variety for other lattice geometries in arbitrary dimension. To remain concrete, we present a couple of two dimensional examples but the reader may readily recognize from the proofs that generalizations are possible (compare Remark~\ref{rem:zdgraph}).
\begin{coro} \label{coro:ehm}
If $A$ denotes the Laplacian on the extended Harper lattice, $q_1$ and $q_2$ are coprime, and $V$ is $q$-periodic, then the Bloch variety of $A+V$ is irreducible modulo periodicity.
\end{coro}
\begin{coro} \label{coro:tri}
If $A$ denotes the Laplacian on the triangular lattice, then for any periodic $V$, the Bloch variety of $A+V$ is irreducible modulo periodicity.
\end{coro}

Note that irreducibility of the Bloch variety is potentially sensitive to modifications in the hopping terms (i.e., the matrix elements of $A$). To the best of our knowledge, even the results of Corollaries \ref{coro:ehm} and \ref{coro:tri} are new. For further details, including definitions of the triangular and extended Harper lattices, see Section~\ref{sec:examples}. To emphasize the distinction between the above models, we present the corresponding polynomials below, recalling that Equation~\eqref{eq:pdefinition} provides the dictionary between $A$ and $P_A$.
\begin{enumerate}[label=(\roman*)]
    \item For the discrete Laplacian on $\bbZ^d$, \begin{equation*}p_{-\Delta}(z)=-\left(z_1+\frac{1}{z_1}+z_2+\frac{1}{z_2}+\cdots+z_d+\frac{1}{z_d}\right) \end{equation*}
    \item For the extended Harper lattice
    \[p_{\rm EHM}(z)=-\left(z_1+\frac{1}{z_1}+z_2+\frac{1}{z_2}+\frac{z_1}{z_2}+\frac{z_2}{z_1}+z_1z_2+\frac{1}{z_1z_2}\right)\]
    \item For the triangular lattice, 
    \begin{equation*}
        p_{\rm tri}(z)=-\left(z_1+\frac{1}{z_1}+z_2+\frac{1}{z_2}+\frac{z_1}{z_2}+\frac{z_2}{z_1}\right).
    \end{equation*}
\end{enumerate}
In particular, in dimension $d=2$, $p_{\rm EHM}(z)$ adds to $p_{-\Delta}(z)$ next nearest neighbor terms and is symmetric with respect to the map $z_j \mapsto z^{-1}_{j}$ for $j=1,2.$ The polynomial $p_{\rm tri}(z)$ does not possess this symmetry, nonetheless the corresponding variety still falls into the scope of Theorem \ref{t:blochIrr}. The triangular lattice is depicted in Figure~\ref{fig:trilat}. Applying a simple shear transformation reduces the triangular lattice to the square lattice with additional edges, as shown in Figure~\ref{fig:trishear}, and hence places the Laplacian on the triangular lattice into the context of the paper after a suitable change of coordinates.
\begin{figure*}[h]
 \begin{minipage}{0.45\textwidth}
 \centering
\begin{tikzpicture}[yscale=.84,xscale=.84]
\filldraw[color=black, fill=black](-3,-3) circle (0.18);
\filldraw[color=black, fill=black](-1,-3) circle (0.18);
\filldraw[color=black, fill=black](1,-3) circle (0.18);
\filldraw[color=black, fill=black](3,-3) circle (0.18);
\filldraw[color=black, fill=black](-2,{sqrt(3)-3}) circle (0.18);
\filldraw[color=black, fill=black](0,{sqrt(3)-3}) circle (0.18);
\filldraw[color=black, fill=black](2,{sqrt(3)-3}) circle (0.18);
\filldraw[color=black, fill=black](-3,{2*sqrt(3)-3}) circle (0.18);
\filldraw[color=black, fill=black](-1,{2*sqrt(3)-3}) circle (0.18);
\filldraw[color=black, fill=black](1,{2*sqrt(3)-3}) circle (0.18);
\filldraw[color=black, fill=black](3,{2*sqrt(3)-3}) circle (0.18);
\filldraw[color=black, fill=black](-2,{3*sqrt(3)-3}) circle (0.18);
\filldraw[color=black, fill=black](0,{3*sqrt(3)-3}) circle (0.18);
\filldraw[color=black, fill=black](2,{3*sqrt(3)-3}) circle (0.18);
\draw [-,line width = .05cm] (-3,-3) -- (3,-3);
\draw [-,line width = .05cm] (-3,{sqrt(3)-3}) -- (3,{sqrt(3)-3});
\draw [-,line width = .05cm] (-3,{2*sqrt(3)-3}) -- (3,{2*sqrt(3)-3});
\draw [-,line width = .05cm] (-3,{3*sqrt(3)-3}) -- (3,{3*sqrt(3)-3});
\draw [-,line width=.05cm] (-3,{2*sqrt(3)-3}) -- (-2,{3*sqrt(3)-3});
\draw [-,line width=.05cm] (-3,-3) -- (0,{3*sqrt(3)-3});
\draw [-,line width=.05cm] (-1,-3) -- (2,{3*sqrt(3)-3});
\draw [-,line width=.05cm] (1,-3) -- (3,{2*sqrt(3)-3});
\draw [-,line width=.05cm] (-3,{2*sqrt(3)-3}) -- (-1,-3);
\draw [-,line width=.05cm] (-2,{3*sqrt(3)-3}) -- (1,-3);
\draw [-,line width=.05cm] (0,{3*sqrt(3)-3}) -- (3,-3);
\draw [-,line width=.05cm] (2,{3*sqrt(3)-3}) -- (3,{2*sqrt(3)-3});
\draw [->,line width=.05cm,color=blue] (-2,{sqrt(3)-3}) -- (-1,{2*sqrt(3)-3});
\draw [->,line width=.06cm,color=blue] (-2,{sqrt(3)-3}) -- (-1,{2*sqrt(3)-3});
\draw [->,line width=.06cm,color=blue] (-2,{sqrt(3)-3}) -- (0,{sqrt(3)-3});
\node [above] at (-2,{sqrt(3)+.3-3}) {\cold{$\bm{b}_2$}};
\node [below] at (-1.3,{sqrt(3)-.1-3}) {\cold{$\bm{b}_1$}};
\end{tikzpicture}
\caption{A portion of the triangular lattice}\label{fig:trilat}
\end{minipage}
\hfill
 \begin{minipage}{0.45\textwidth}
 \centering
\begin{tikzpicture}[yscale=.75,xscale=.75]
\filldraw[color=black, fill=black](-3,-3) circle (0.18);
\filldraw[color=black, fill=black](-1,-3) circle (0.18);
\filldraw[color=black, fill=black](1,-3) circle (0.18);
\filldraw[color=black, fill=black](3,-3) circle (0.18);
\filldraw[color=black, fill=black](-3,-1) circle (0.18);
\filldraw[color=black, fill=black](-1,-1) circle (0.18);
\filldraw[color=black, fill=black](1,-1) circle (0.18);
\filldraw[color=black, fill=black](3,-1) circle (0.18);
\filldraw[color=black, fill=black](-3,1) circle (0.18);
\filldraw[color=black, fill=black](-1,1) circle (0.18);
\filldraw[color=black, fill=black](1,1) circle (0.18);
\filldraw[color=black, fill=black](3,1) circle (0.18);
\filldraw[color=black, fill=black](-3,3) circle (0.18);
\filldraw[color=black, fill=black](-1,3) circle (0.18);
\filldraw[color=black, fill=black](1,3) circle (0.18);
\filldraw[color=black, fill=black](3,3) circle (0.18);
\draw [-,line width = .06cm] (-3,-3) -- (3,-3);
\draw [-,line width = .06cm] (-3,-3) -- (-3,3);
\draw [-,line width=.06cm] (-3,-1) -- (3,-1);
\draw [-,line width = .06cm] (-1,-3) -- (-1,3);
\draw [-,line width=.06cm] (-3,1) -- (3,1);
\draw [-,line width = .06cm] (1,-3) -- (1,3);
\draw [-,line width=.06cm] (-3,3) -- (3,3);
\draw [-,line width = .06cm] (3,-3) -- (3,3);
\draw [-,line width = .06cm] (-3,3) -- (3,-3);
\draw [-,line width = .06cm] (-1,-3) -- (-3,-1);
\draw [-,line width = .06cm] (1,-3) -- (-3,1);
\draw [-,line width = .06cm] (-1,3) -- (3,-1);
\draw [-,line width = .06cm] (3,1) -- (1,3);
\end{tikzpicture}
\caption{The triangular lattice after shearing.}\label{fig:trishear}
\end{minipage}
\end{figure*}

\subsection{Definitions and Context}\label{sec:definitions}
Let us now give relevant definitions and context. Given $q_i\in \bbN$, $i=1,2,\ldots,d$,
let $\Gamma = \Gamma_q :=q_1\Z\oplus q_2 \Z\oplus\cdots\oplus q_d\Z$.
We say that a function $V: \Z^d\to \C$ is  $q$-periodic ($\Gamma$-periodic, or just periodic)  if $V_{n+\gamma} = V_n$ for all $n \in \bbZ^d$ and all $\gamma\in \Gamma$. 
\begin{definition} Let $\bbC^\star = \bbC\setminus \{0\}$. For $z = (z_1,\ldots,z_d) \in (\bbC^\star)^d$ and $q = (q_1,\ldots,q_d) \in \bbN^d$, the space $\scrH(z,q)$ consists of those $\psi :\bbZ^d \to \bbC$ for which
\begin{equation}
\psi_{n+j\odot q} = z^j \psi_n, \forall n,j\in \bbZ^d,
\end{equation}
where we write $j\odot q=(j_1q_1,\ldots,j_d q_d)$ and use the multi-index notation $z^j = z_1^{j_1} \cdots z_d^{j_d}$. Naturally, $\scrH(z,q)$ is a Hilbert space of finite dimension $Q:=q_1\cdots q_d$.

If $V:\bbZ^d \to \bbC$ is $q$-periodic, the corresponding Bloch variety is given by 
\begin{equation} \label{eq:blochvarDef}
B=B(H) = \{(k,\lambda) \in \bbC^{d+1} : H\psi = \lambda \psi \text{ enjoys a nonzero solution in } \scrH(e^{2\pi i k},q)\},
\end{equation}
where we write $e^{2\pi i k} = (e^{2\pi i k_1},\ldots,e^{2\pi i k_d}) \in (\bbC^\star)^d$.
We employ here a standard abuse of notation in which $H$ represents both the self-adjoint operator in $\ell^2(\bbZ^d)$ and the difference operator acting in, say, $\ell^\infty(\bbZ^d)$.
\end{definition} 

\begin{definition} We will say the Bloch variety $B(H)$ is \emph{irreducible modulo periodicity} if for every two irreducible components $\Omega_1$ and $\Omega_2$ of $B(H)$, there exists $m \in \mathbb{Z}^d$ such that $\Omega_1 = (m,0)+\Omega_2$. \end{definition}

\begin{definition}
Given  $\lambda\in \C$, the  Fermi surface (variety) $F_{\lambda}(H)$ is defined as the level set of the  Bloch variety:
\begin{equation*}
F_{\lambda}(H)=\{k\in\C^d: (k,\lambda)\in B(H)\}.
\end{equation*}
\end{definition}

We should mention that 	reducible Fermi and Bloch  varieties  are known to occur for  periodic graph operators,  e.g.,  ~\cite{shi2, fls}. 
 One challenging problem in the study of periodic operators is to prove the (ir)reducibility of the Bloch and  Fermi  varieties ~\cite{GKTBook, ktcmh90, bktcm91, bat1, batcmh92, ls, shi2, fls, GKToverview,shva}.
For instance, irreducibility of the Bloch variety implies that in case $B(H)\cap U\neq \emptyset$ for some open set $U\subset \bbC^{d+1}$, the knowledge of $B(H)\cap U$ allows one to recover $B(H)$.
Besides its own importance in algebraic geometry, the
(ir)reducibility of these varieties is crucial  in the study of spectral properties of periodic elliptic operators. 
In particular, this has implications for the structure of spectral  band edges~\cite{LiuPreprint:Irreducibility,fs22}, the isospectrality~\cite{LiuPreprint:fermi} and the  existence of embedded eigenvalues for operators perturbed by a local defect~\cite{kv06cmp, kvcpde20, shi1, IM14, AIM16,dks,liu2021topics}. 
Based on existing evidence, Kuchment conjectures that the Bloch variety of any periodic second-order elliptic operator is irreducible \cite[Conjecture 5.17]{Kuchment2016BAMS}.

 There have been many works that address irreducibility of the Bloch and Fermi varieties (see, e.g.,~\cite{bat1, batcmh92, battig1988toroidal, bktcm91, GKTBook, ktcmh90, LiuPreprint:Irreducibility}).
In two dimensions, B{\"a}ttig  \cite{battig1988toroidal} showed that the Bloch variety  $B(-\Delta+V)$  is irreducible (modulo periodicity). In \cite{GKTBook},  Gieseker, Kn\"orrer and Trubowitz proved  that  $F_{\lambda}(-\Delta+V)/\Z^2$ is irreducible except for finitely many values of  $\lambda$. When $d=3$, irreducibility of $F_{\lambda}(-\Delta+V)/\Z^d$ for any $\lambda\in \C$ was proved by B{\"a}ttig  \cite{batcmh92}.  In a recent paper of the second author~\cite{LiuPreprint:Irreducibility} it is showed that when $d\geq3$ $F_{\lambda}(-\Delta+V)/\Z^d$ is irreducible for any $\lambda\in \C$ and when $d=2$ $F_{\lambda}(-\Delta+V)/\Z^2$  is irreducible for all $\lambda\in \mathbb{C}\setminus [V]$, where $[V]$ is the average of $V$ over one periodicity cell. It follows from these works that when $d\geq2$ the Bloch variety  $B(-\Delta+V)$  is irreducible (modulo periodicity). For continuous periodic Schr\"odinger operators, when $d=2$ Kn\"orrer and Trubowitz \cite{ktcmh90} proved  that  the Bloch variety is irreducible  (modulo periodicity) and for $d=3$,  B{\"a}ttig,  Kn\"orrer and Trubowitz proved that the Fermi variety  at any level is irreducible (modulo periodicity) for separable periodic  potentials \cite{bktcm91}.
In \cite{ls}, irreducibility of the Fermi variety for all but finitely many energies is proved for a suitable class of planar periodic graphs.

In  \cite{GKTBook, ktcmh90, bktcm91, bat1, batcmh92, battig1988toroidal}, algebraic geometry techniques are employed to construct the toroidal and directional compactifications of Fermi and  Bloch varieties and understand asymptotics of their defining (Laurent) polynomials. The perspective employed by us  in the current manuscript is inspired by \cite{LiuPreprint:Irreducibility}.
In general terms, the goal is to explicitly calculate asymptotics of the  (Laurent) polynomials at $z\in\{z: z_j=0 \text{ or } z_j=\infty, j=1,2, \cdots,k\}$ and
	show that these asymptotics  contain enough information about the original variety. Concretely, the proof is based on  changing variables,  studying of the lowest degree components of a family of (Laurent) polynomials in several variables and  degree arguments. With regards to the Bloch variety, we expand the approach of \cite{LiuPreprint:Irreducibility} in different directions. As a consequence, for the main result of Theorem~\ref{mainthm} below, the underlying lattice may be of very general nature and contain somewhat arbitrary finite-range connections (see ~\ref{assump1} and ~\ref{assump2} below for a more precise statement of the assumptions). In particular, we obtain irreducibility for the Bloch variety corresponding to periodic Schr\"{o}dinger operators on the triangular lattice and the extended Harper lattice; see Section~\ref{sec:examples} for a precise description of these examples.
While our approach is inspired by \cite{LiuPreprint:Irreducibility}, we do not follow the same path. By working directly with the lowest degree components, we can eschew a discussion of asymptotic statements about the varieties themselves.

The structure of the paper is as follows. We precisely formulate Theorem~\ref{mainthm}, our main   result, in Section~\ref{sec:mainresult}. Section~\ref{sec:technical lemmas} contains preparatory technical results that are then employed in Section~\ref{sec:mainproof} to prove Theorem~\ref{mainthm}. We elucidate the connection between this result and periodic operators in Section~\ref{sec:floquet}, which also contains some relevant background on periodic long-range Schr\"odinger operators.  We conclude in Section~\ref{sec:examples} with the proof of  Theorem~\ref{t:blochIrr} and some relevant examples and applications.

\subsection*{Acknowledgements}
We are grateful to the anonymous reviewer for carefully reading the manuscript and for helpful suggestions that improved the paper.

\section{Main Result} \label{sec:mainresult}
To state the main result, we begin by recalling some crucial terminology.
\begin{definition}
	Suppose $f$ is a Laurent monomial in $m$ variables, that is, $f(z) =cz^\alpha= cz_1^{\alpha_1}z_2^{\alpha_2}\cdots z_m^{\alpha_m}$ with $\alpha_i\in\Z$ for $i=1, \ldots, m$ and $c \neq 0$. The degree of $f$ is defined as
	$\deg (f)=\alpha_1+\alpha_2+\cdots +\alpha_m$. 
	Abusing notation slightly, we also denote $\deg(\alpha) = \alpha_1 + \alpha_2 + \cdots +\alpha_m$ for the multi-index $\alpha = (\alpha_1,\ldots,\alpha_m) \in \bbZ^m$.
\end{definition}

\begin{definition}\label{lowestdfn}
	Given a Laurent polynomial
	\[p(z)= \sum c_\alpha z^{\alpha},\]
	let $L_- = \min\{\deg(\alpha) : c_\alpha \neq 0\}$. Then,  the \emph{lowest degree component} of $p$ is defined to be the Laurent polynomial 
	\[h(z)=\sum_{\deg \alpha =L_{-}}c_\alpha z^{\alpha}.\]
\end{definition}
One of the crucial properties of this notion is the following: denoting the lowest-degree component of $p$ by $\underline{p}$, one has $\underline{(fg)} = \underline{f}\cdot\underline{g}$, which enables one to relate factorizations of a polynomial to factorizations of its lowest-degree component. Obviously, some care is needed to deduce nontrivial consequences from this observation in the context of our main result.

Let us write $\bbC[z_1,\ldots,z_m]=:\bbC[z]$ for the set of polynomials in $z_1,\ldots,z_m$. Similarly, we write
$\bbC[z_1,z_1^{-1},\ldots,z_m,z_m^{-1}]=:\bbC[z,z^{-1}]$ for the set of Laurent polynomials\footnote{This involves a minor, albeit common abuse of notation, since one has the relation $z_jz_j^{-1}=1$ in $\bbC[z,z^{-1}]$.} in $z_1,\ldots,z_m$.

\begin{definition}
	
	Recall that a polynomial $\mathcal{P} \in \bbC[z]$ is called reducible if there exist nonconstant polynomials $f,g \in \bbC[z]$ such that $\mathcal{P}=fg$ and \emph{irreducible} otherwise. Similarly, we say  that  a Laurent polynomial $\mathcal{P} \in \bbC[z,z^{-1}]$ is irreducible if it  can not be factorized non-trivially, that is, there are no non-monomial Laurent polynomials $f,g$ such that $\mathcal{P}= fg$. 
	
	Notice that nonconstant monomials are units in the algebra of Laurent polynomials, which accounts for a small subtlety. That is, one must be somewhat careful here with zeros at $z=0$ and $z = \infty$. The polynomial $z^2$ is reducible in $\bbC[z]$ but is a unit in $\bbC[z,z^{-1}]$. In practice, this should cause no confusion, and we will write that $\mathcal{P}$ is irreducible in $\bbC[z]$ (respectively in $\bbC[z,z^{-1}]$) if we wish to emphasize the sense in which irreducibility is meant in a specific context.
\end{definition}

\begin{remark}
	If $P$ is an irreducible Laurent polynomial in $m$ variables, then the corresponding variety $\{z\in (\C^{\star})^m: \mathcal{P}(z)=0\}$ is irreducible as an analytic set.\footnote{The converse is clearly false, which may be considered by considering the variety associated with $f(z)^2$, where $f(z)$ is irreducible. This issue can be elegantly resolved using the language of schemes \cite{Hartshorne1977GTM}.}
Thus, the overall strategy of our work is to show that a suitable Laurent polynomial that describes the Bloch variety is irreducible.
Concretely, we may consider the set $\mathcal{B}(H)$ which consists of those $(z,\lambda)\in (\bbC^\star)^d \times \bbC$ such that $H\psi = \lambda \psi$ enjoys a nontrivial solution $\psi \in \scrH(z,q)$.
By Floquet theory, one may determine a suitable Laurent polynomial ${\mathcal{P}}(z,\lambda)$ such that $\mathcal{B}(H)$ is precisely the zero set of ${\mathcal{P}}$ (see Section~\ref{sec:floquet}).
Thus, since $(k,\lambda) \in B(H)$ if and only if $(e^{2\pi i k},\lambda) \in \mathcal{B}(H)$, to show that $B(H)$ is irreducible modulo periodicity, it suffices to show that the corresponding Laurent polynomial is irreducible.
	
\end{remark}

Let us begin by collecting some notation that we will use throughout the paper.
Given $q=(q_1,\ldots,q_d) \in\bbN^d$, we define the lattice $\Gamma$ by
\begin{equation}
\Gamma = \bigoplus_{j=1}^d q_j\bbZ= \{ n \in \bbZ^d : q_j | n_j \ \forall\,\, 1\le j \le d\}
\end{equation}
and
 the fundamental cell, $W$,  by
	\begin{equation} \label{eq:fundcelldef}
W=\{n=(n_1,n_2,\ldots,n_d)\in\bbZ^d: 0\leq n_j\leq q_{j}-1, j=1,2,\ldots, d\} = \bbZ^d \cap \prod_{j=1}^d [0,q_j).
\end{equation}
Given $n\in W$ and $j\in\{1,\ldots,d\}$, let
\begin{equation}\label{action}
 \mu^{j}_{n}=e^{2\pi i\frac{n_j}{q_j}}. 
\end{equation}
and denote by $\mu_n$ the vector $(\mu^{1}_{n},\ldots, \mu^{d}_{n})$. We also let
\begin{equation} \label{eq:characterActionDef}
\mu_n \odot\left(z_{1}, z_{2},\ldots, z_d\right)=\left(\mu^{1}_{n}z_{1}, \mu^{2}_{n}z_{2},\ldots,\mu^{d}_{n}z_d\right).
\end{equation}
Let $p$ be a Laurent polynomial and define
\begin{equation}
p_{n}(z)=p(\mu_n\odot z), \quad n \in W, \ z \in (\bbC^{\star})^d.
\end{equation}

We shall work with Laurent polynomials in $m=d+1$ variables $z_1,...,z_d,\lambda$. Abusing notation somewhat, we write $\bbC[z,\lambda]$ (respectively $\bbC[z,\lambda,z^{-1},\lambda^{-1}]$) for the set of polynomials (respectively the set of Laurent polynomials) in $z$ and $\lambda$. The Laurent polynomials of interest in the present work are those of the form 
\begin{equation}\label{mathcalP}
\widetilde{\mathcal{P}}(z,\lambda)=\prod_{n\in W} (p_{n}(z)-\lambda)+\sum_{X\in \mathcal{S}}C_{X}\prod_ { n\in X}(p_n(z)-\lambda),\end{equation}
where the summation runs over $X$ in an arbitrary collection $\mathcal{S}$ of proper subsets of $W$ and $C_{X}\in \bbC$.  In fact, $\widetilde{\mathcal{P}}(z,\lambda)$ is a Laurent polynomial in the variable $z$ and a polynomial in $\lambda$.
Collecting terms, we see that 
\begin{equation} \label{eq:PmonicInlambda}
\widetilde{\mathcal{P}}(z,\lambda) = (-1)^Q \lambda^Q + \sum_{k=0}^{Q-1}b_k(z)\lambda^k,\end{equation}
where $b_k \in \bbC[z,z^{-1}]$ and $Q=q_1\cdots q_d$.

Note that we do not exclude the case $\emptyset \in\mathcal{S}$, our convention being that $\prod_ { n\in \emptyset}(p_n(z) - \lambda)=1$.
These are exactly the types of polynomials that one produces by expanding the determinant of the Floquet operator associated with a suitable periodic operator, hence their interest in the current work.

 For each $X$, the constant $C_{X}$ is assumed to be independent of $\lambda$ and $z$. Assume further that $\widetilde{\mathcal{P}}(z,\lambda)$ is invariant under action of each $\mu_n$, i.e., 
\begin{equation} \label{eq:mathcalPCharInv}
\widetilde{\mathcal{P}}(z,\lambda)=\widetilde{\mathcal{P}}(\mu_n\odot z,\lambda) \text{ for all } n\in W.
\end{equation}

\begin{remark} \label{rem:polyToSchrod}
The assumptions \eqref{mathcalP} and \eqref{eq:mathcalPCharInv} include the central example 
where
\begin{equation*}
 \widetilde{\mathcal{P}}(z,\lambda)=\mathrm{det}\left(D+B-\lambda I \right)  
\end{equation*} and the matrices $D=D(z)$ and $B$ are defined by
\begin{equation} \label{eq:rem21Adef}
	D(n,n')=p_n(z) \delta_{n,n'}\
\end{equation}
\begin{equation}
 B(n,n')=\widehat V\left(\frac{n_1-n'_1}{q_1},\ldots,\frac{n_d-n'_{d}}{q_d}\right),\,\,\, n,n'\in W.
\end{equation}
Here $\widehat V$ denotes the discrete Fourier transform of $V$, defined as in \eqref{eq:discretefourier}.
For further discussion, see Section~\ref{sec:floquet}, especially Proposition~\ref{prop:floquetTransf}.

Let us note the key properties are that $D$ is a diagonal matrix and the entries of $B$ are independent of $z$. Consequently, neither self-adjointness of $A$ or real-valuedness of $V$ is a crucial ingredient.
\end{remark}
Since $\widetilde{\mathcal{P}}(z,\lambda)$ is invariant under the action of each $\mu_n$, it is elementary to check (cf.\ Lemma~\ref{liftlemma}) that there exists $\mathcal{P}(z,\lambda)$ such that
\begin{equation}\label{Pdef}\widetilde{\mathcal{P}}(z,\lambda)=\mathcal{P}(z_1^{q_1},z_2^{q_2},\ldots,z_d^{q_d},\lambda).
\end{equation}

Our goal is to show that $\mathcal{P}(z,\lambda)$ is irreducible as a Laurent polynomial under the assumptions below. 
\begin{enumerate}[label=(\subscript{A}{{\arabic*}})]
		\item\label{assump1} $\deg(h) < 0$,  where $h$ denotes the lowest degree  component of $p$, (see Definition~\ref{lowestdfn}). 
		\item\label{assump2}
		 Letting $h_n(z)=h(\mu_n \odot z)$ , the polynomials $\{h_n(z)\}_{n\in W}$ are pairwise distinct.
	\end{enumerate}
	
The reader may readily check that $p_{n+m}(z) = p_n(\mu_m \odot z)$ (with addition of indices computed mod $\Gamma$). Thus, to check Assumption~\ref{assump2} in practice, it suffices to show that $h_0 \neq h_n$ for every $n \in W\setminus\{0\}$.

\begin{theorem}\label{mainthm}
Let $p \in \bbC[z,z^{-1}]$, $q \in \bbN^d$, $\mathcal{S}$ a collection of proper subsets of $W$, and complex numbers $\{C_X\}_{X \in \mathcal{S}}$ be given. Assume that $\widetilde{\mathcal{P}}$ is a polynomial of the form \eqref{mathcalP} obeying \eqref{eq:mathcalPCharInv}, and let $\mathcal{P}$ be the polynomial given by \eqref{Pdef}.
Under Assumptions~{\rm\ref{assump1}} and {\rm\ref{assump2}}, we conclude that $\mathcal{P}$ is irreducible as a Laurent polynomial. 
\end{theorem}

As mentioned in Remark~\ref{rem:polyToSchrod}, the connection to Schr\"odinger operators and Theorem \ref{t:blochIrr} will be established in Section~\ref{sec:floquet}.

\begin{remark}\label{rem:notation} Let us collect some notation from the previous paragraphs that will be repeatedly used throughout the proofs.
\begin{enumerate}
    \item $\bbC[z]$ (resp.\ $\bbC[z,z^{-1}]$) denotes the set of polynomials (resp.\ Laurent polynomials) in $z_1,\ldots,z_d$.\smallskip
    
\item $p \in \bbC[z,z^{-1}]$.\smallskip

\item $h(z)$ is the lowest degree component of $p(z)$.\smallskip
\item $\Gamma = q_1\bbZ \oplus \cdots \oplus q_d \bbZ^d$, $W=\bbZ^d \cap \prod_{j=1}^d [0,q_j)$, $\mathcal{S} \subset 2^{W}\setminus\{W\}$ is arbitrary.\smallskip
\item $\mu^{j}_{n}=e^{2\pi in_j/q_j}$, $n \in \bbZ^d$, $j=1,\cdots,d$.\smallskip

\item For $n \in W$, $\mu_n=(\mu^{1}_{n},\ldots,\mu^{d}_{n})$.
\smallskip
\item $p_n(z)=p(\mu_n \odot z).$\smallskip
\item $\widetilde{\mathcal{P}}(z,\lambda)$ is given by \begin{equation*}\widetilde{\mathcal{P}}(z,\lambda)=\prod_{n\in W} (p_{n}(z)-\lambda)+\sum_{X\in \mathcal{S}}C_{X}\prod_ { n\in X}(p_n(z)-\lambda),\end{equation*}\smallskip
\item $Q=q_1\cdots q_d$
\smallskip

\item $z^\alpha = z_1^{\alpha_1} \cdots z_d^{\alpha_d}$ for $z \in (\bbC^\star)^d$, $\alpha \in \bbZ^d$.\smallskip

    \item $\mathcal{P}(z,\lambda)$ is defined by 
 \[\widetilde{\mathcal{P}}(z,\lambda)=\mathcal{P}(z_1^{q_1},z_2^{q_2},\ldots,z_d^{q_d},\lambda).\]  
    \item $a\odot b = (a_1b_1,\ldots,a_db_d)$ for ordered $d$-tuples $a=(a_1,\ldots,a_d)$ and $b = (b_1,\ldots,b_d)$.
\end{enumerate}
\end{remark}

\section{Technical Lemmas}
\label{sec:technical lemmas}
\begin{lemma}\label{liftlemma} Suppose $\widetilde{g}$ is a Laurent polynomial in $z$ and $\lambda$. With notation as in Remark~\ref{rem:notation}, one has that $\widetilde{g}(z,\lambda) \equiv \widetilde{g}(\mu_n\odot z,\lambda)$ for every $n \in W$ if and only if there is a Laurent polynomial $g(w,\lambda)$ such that
\begin{equation} \label{eq:fromPtocalP}
\widetilde{g}(z,\lambda) 
\equiv g(z_1^{q_1},\ldots,z_d^{q_d},\lambda).
\end{equation}
\end{lemma}
\begin{proof}
 If $g(w,\lambda)$ satisfying \eqref{eq:fromPtocalP} exists, then it readily follows from the definition of $\mu_n\odot z$ that 
 \[
\widetilde{g}(\mu_n\odot z,\lambda) 
 = g\left((\mu_n^1z_1)^{q_1}, \ldots, (\mu_n^dz_d)^{q_d}, \lambda\right)
 = g\left(z_1^{q_1}, \ldots, z_d^{q_d},\lambda\right)
  = \widetilde{g}(z,\lambda).
 \] To see that the converse implication holds, write \[\widetilde{g}(z,\lambda) = \sum_{\ell \in \bbZ^d,m \in \bbZ} \widetilde{c}_{\ell,m}z^\ell\lambda^m \text{ and let } g(w,\lambda) = \sum_{\ell \in \bbZ^d,m \in \bbZ} c_{\ell,m}w^\ell\lambda^m\] be another Laurent polynomial. Note that 
 \[g(z_1^{q_1},\ldots,z_d^{q_d},\lambda) = \sum_{\ell \in \bbZ^d,m \in \bbZ} c_{\ell,m}z^{\ell\odot q}\lambda^m.\] Thus \eqref{eq:fromPtocalP} holds if and only if for all $m\in \mathbb{Z}$  \[  \widetilde{c}_{\ell,m} = \begin{cases} c_{r,m} & \ell = q\odot r \in \Gamma, \\ 0 & \text{otherwise,}\end{cases}\]  and hence $g$ satisfying \eqref{eq:fromPtocalP} exists if and only if $\widetilde{c}_{\ell,m}=0$ for all $\ell \notin \Gamma$ and $m\in \mathbb{Z}$. Thus, if \eqref{eq:fromPtocalP} does not hold, we must have $\widetilde{c}_{\ell,m}\neq 0$ for some $\ell \notin \Gamma$, say $q_i \! \not|\ \ell_i$ for some $i\in\{1,\cdots,d\}$.  Choose $n=e_i \in W$, and note that for this choice of $n$ one has\[ \widetilde{g}(z,\lambda) - \widetilde{g}(\mu_n\odot z,\lambda) \not\equiv 0, \] concluding the proof.
 \end{proof}
\begin{definition}\label{gammadfn'}
	For each $j\in \{1,2,\cdots,d\}$, define $\gamma_j' \geq 0$ as follows. We let $-\gamma_{j}'$ be the lowest exponent of $z_j$ in $h(z)$ in case this exponent is negative and $\gamma_j'=0$ otherwise.
\end{definition}
\begin{lemma}\label{rnirred} Let $p$ be a Laurent polynomial in $z_1,\ldots,z_d$ and let $h$ be  the lowest degree component of $p$. 
 Then, the polynomials 
 \[
 r_n(z,\tlambda)=\widetilde{\lambda} z^{\gamma_{1}'}_{1}\cdots z^{\gamma_{d}'}_{d}h(\mu_n\odot z)-z^{\gamma_{1}'}_{1}\cdots z^{\gamma_{d}'}_{d}
 \]
are irreducible in $\bbC[z,\widetilde\lambda]$ for each $n\in W$.
 Moreover, under Assumption {\rm\ref{assump2}}, we conclude that for any $n \neq n' \in W$, $r_n$ and $r_{n'}$ are relatively prime.
\end{lemma}
\begin{proof}
Assume for the sake of contradiction that  $r_n(z,\tlambda)$ is reducible. Since the degree of $\tlambda$ in $r_n(z,\tlambda)$ is one, we must have that
\begin{equation}\label{rirred}r_n(z,\tlambda)=f(z,\tlambda)g(z)
\end{equation}
for non-constant polynomials $f(z,\tlambda)$ and $g(z)$. 
Since $\tlambda$ does not divide $r_n(z,\tlambda)$  in $\bbC[z]$, we see that there exist non-zero polynomials $f_1(z)$ and $f_2(z)$ such that 
\[f(z,\tlambda)=\tlambda f_1 (z) - f_2(z).\]
From \eqref{rirred} and the definition of $r_n(z,\tlambda)$ we obtain $f_2(z)g(z) = z^{\gamma_{1}'}_{1}\cdots z^{\gamma_{d}'}_{d}$. 
In particular, $g(z)= z^{m_1}_{1}\cdots z^{m_d}_{d}$ where $m_1,\ldots,m_d$ are integers with $0\leq m_j\leq \gamma_j'$ for $j\in\{1,\ldots,d\}$. 
Since $g$ is nonconstant, $m_l >0$ for at least one $l$. In particular, $$\gamma_l' \geq m_l>0.$$ Consequently, \eqref{rirred} implies that the polynomial $z^{\gamma_{1}'}_{1}\cdots z^{\gamma_{d}'}_{d}h(\mu_n\odot z)$ is divisible by $z_l$ for some $l\in \{1,2,\ldots,d\}.$ However, the lowest degree of $z_l$ in $h(\mu_n\odot z)$ is, by definition, equal to $-\gamma_{l}'$. Thus $z^{\gamma_{1}'}_{1}\cdots z^{\gamma_{d}'}_{d}h(\mu_n\odot z)$ is not divisible by $z_l$, contradicting \eqref{rirred}. 
Consequently, $r_n$ is irreducible.

To prove the second statement of the lemma, assume $r_n$ and $r_{n'}$ share a nontrivial common factor.  By irreducibility, they must be constant multiples of one another. However, from the definition, this is only possible if $r_n=r_{n'}$, which contradicts Assumption~\ref{assump2}.
\end{proof}

Let us introduce the auxiliary polynomial
\begin{equation} \label{eq:tildeADef}
\widetilde{a}(z,\tlambda)=\prod_{n\in W}r_n(z,\tlambda)    
\end{equation}
with $r_n(z,\tlambda)$ as in Lemma~\ref{rnirred}  for $n\in W$.
By a direct calculation, $\widetilde{a}(z,\tlambda)$ is invariant under the replacement\footnote{Later on, we will call this the \emph{action} of $\mu_n$ on a polynomial.} $z \mapsto \mu_n \odot z$, so, as a consequence of Lemma~\ref{liftlemma}, there exists $a(z,\tlambda)$ such that
\begin{equation}\label{tiladfn}
\widetilde{a}(z,\tlambda)=a(z_1^{q_1},\ldots,z_d^{q_d},\tlambda).
\end{equation}

\begin{lemma}\label{airred}
	Under Assumption~{\rm\ref{assump2}}, the   polynomial $a(z,\tlambda)$ given by \eqref{tiladfn} is irreducible in $\bbC[z,\tlambda]$.
\end{lemma}

\begin{remark}
It is important that we pass to the lift $a$ here, since $\widetilde a$ is clearly reducible. 
\end{remark}

\begin{proof}[Proof of Lemma~\ref{airred}]
	Suppose for the sake of establishing a contradiction that $a(z,\tlambda)$ is reducible, and write
	\begin{equation}\label{areducible}
	   a(z,\tlambda)={f}_1(z,\tlambda){g}_1(z,\tlambda)
	\end{equation} for non-constant polynomials $f_1$ and $g_1$. Let $\widetilde{f}_1(z,\tlambda)={f}_1(z_1^{q_1},\ldots,z_d^{q_d},\tlambda)$ and $\widetilde{g}_1(z,\tlambda)={g}_1(z_1^{q_1},\ldots,z_d^{q_d},\tlambda)$. Combining \eqref{tiladfn} and \eqref{areducible} yields 
	\[\widetilde {a}(z,\tlambda)=\widetilde{f}_1(z,\tlambda)\widetilde{g}_1(z,\tlambda).\] 
	 Moreover, by definition $\widetilde{f}_1(z,\tlambda)$ and $\widetilde{g}_1(z,\tlambda)$ are both invariant under the action of each $\mu_n$. 
	 Recall from Lemma~\ref{rnirred} that each $r_n(z,\tlambda)$ is irreducible. Therefore, each $r_n(z,\tlambda)$ is a factor of either $\widetilde{f}_1$ or $\widetilde{g}_1$. 
	 By invariance of $\widetilde{f}_1(z,\tlambda)$ (respectively $\widetilde{g}_1(z,\tlambda)$) under the action of each $\mu_n$ and since, by Lemma ~\ref{rnirred}, $r_n$ and $r_{n'}$ are relatively prime for $n\neq n'$, we conclude the following: if $\widetilde{f}_1(z,\tlambda)$ (respectively $\widetilde{g}_1(z,\tlambda)$) has a factor of $r_n(z,\tlambda)$ then it must have a factor of
	\[ \prod_{n\in W}r_n(z,\tlambda)=\widetilde{a}(z,\tlambda).\] 
	However, this, together with \eqref{areducible}, implies that either $\widetilde{f}_1(z,\tlambda)$ or $\widetilde{g}_1(z,\tlambda)$ must be constant, which is a contradiction. Thus, we conclude that ${a}(z,\tlambda)$ is irreducible.
\end{proof}

\begin{lemma}\label{Lemma:irredmeets}
 Let $\mathcal{P}(z,\lambda)$ be given by \eqref{Pdef} and let $f$ be any irreducible factor of $\mathcal{P}$. 
 Then
 $f$ must depend on   $ \lambda.$
\end{lemma}
\begin{proof}
  If $f$ is an irreducible factor of $\mathcal{P}$, then $f$ must depend on $\lambda$ since otherwise there would be a suitable choice of $z=(z_1,\ldots,z_d)$,  namely any solution of $f(z)=0$, for which $\mathcal{P}(z,\lambda)=0$ for any $\lambda$. This, in turn, contradicts the fact that the term of highest degree of $\lambda$ in $\mathcal{P}(z,\lambda)$ is  ${(-1)^Q}\lambda^Q$ (see \eqref{eq:PmonicInlambda} and \eqref{Pdef}).
\end{proof}
\section{Proof of Theorem~\ref{mainthm}}
\label{sec:mainproof}
Before proceeding with the proof of the main result, Theorem~\ref{mainthm}, let us introduce some notation.
\begin{definition}\label{gammadfn}
	For each $j\in \{1,2,\ldots,d\}$ denote by $-\gamma_{j}$ the lowest exponent of $z_j$ in $p(z)$ in case this exponent is negative and $\gamma_j=0$ otherwise. Clearly, $\gamma_j \geq \gamma_j'$ with $\gamma_j'$ given in Definition~\ref{gammadfn'}.
\end{definition}

\begin{proof}[Proof of Theorem~\ref{mainthm}]
 Let $\tlambda=\lambda^{-1}$. Then $\mathcal{P}(z,\lambda)=\mathcal{P}(z,\tlambda^{-1})$ is a Laurent polynomial in the variables $(z,\tlambda)$. 
 Let $\gamma_j$, $j=1,\ldots,d$ be as in Definition~\ref{gammadfn}. In case $\gamma_j>0$ for some $j\in\{1,\ldots,d\}$,
	 the lowest power of $z_j$ in  $\mathcal{P}(z,\tlambda^{-1})$ is $-\gamma_jQ/q_j$. 
	 
	  Moreover, the lowest power of $\tlambda$ in $\mathcal{P}(z,\tlambda^{-1})$ is $\tlambda^{-Q}$ (cf.\ \eqref{eq:PmonicInlambda}), so
	 \begin{equation}
	 \mathcal{R}(z,\tlambda)
	 = \left(\widetilde{\lambda}z^{\frac{\gamma_1}{q_1}}_{1}\cdots z^{\frac{\gamma_d}{q_d}}_{d}\right)^Q\mathcal{P}(z,\tlambda^{-1})  \end{equation}
	 defines a polynomial $\mathcal{R} \in \bbC[z,\tlambda]$.
\begin{claim}\label{cl:mainproof}
	For each $1 \le j \le d$, $z_j$ does not divide $\mathcal{R}(z,\tlambda)$.
\end{claim}
\begin{claimproof}
Indeed, if $\gamma_j>0$, this is clear from the definitions, since $-\gamma_j$ is the smallest power of $z_j$ in $p$ and hence $-\gamma_j Q/q_j$ is the smallest power of $z_j$ in $\mathcal{P}$.
Otherwise, $\gamma_j = 0$, and the claim can be seen from \eqref{eq:PmonicInlambda}.\end{claimproof}

	 Since $\tlambda$ also does not divide $\mathcal{R}(z,\tlambda)$, Claim~\ref{cl:mainproof} implies that reducibility of the Laurent polynomial $\mathcal{P}(z,\tlambda^{-1})$ is equivalent to reducibility of the polynomial $\mathcal{R}(z,\tlambda)$.
	 
	 Now, assume for the sake of contradiction that $\mathcal{P}(z,\tlambda^{-1})$ is reducible.
	 There exist $m>1$ and non-constant polynomials $f_{l}(z,\tlambda)$, $l=1,2,\ldots,m$, in  $\bbC[z,\tlambda]$ such that
	   \begin{equation}\label{contradictionsetup}
	  \left(\widetilde{\lambda}z^{\frac{\gamma_1}{q_1}}_{1}\cdots z^{\frac{\gamma_d}{q_d}}_{d}\right)^Q\mathcal{P}(z,\tlambda^{-1})=\prod_{l=1}^mf_l(z,\tlambda).\end{equation}
Let us recall the auxiliary polynomial  $\widetilde{a}$ given by
\[\widetilde{a}(z,\tlambda):=\left(\widetilde{\lambda}z^{\gamma_{1}'}_{1}\cdots z^{\gamma_{d}'}_{d}\right)^Q\prod_{n\in W}(h(\mu_n\odot z) -{\widetilde{\lambda}}^{-1}).\]

	    Let $\tf_l(z,\tlambda)=f_l(z_1^{q_1},\ldots,z_d^{q_d},\tlambda)$.
	    Then, by \eqref{Pdef} and \eqref{contradictionsetup}, we have that
	   \begin{equation}\label{contradicteq}
	       \left(\widetilde{\lambda}z^{\gamma_{1}}\cdots z^{\gamma_{d}}\right)^Q\widetilde{\mathcal{P}}(z,\tlambda^{-1})=\prod_{l=1}^m\tf_l(z,\tlambda).
	   \end{equation}
	   By definition of $\widetilde{\mathcal{P}}$ in (\ref{mathcalP}) one sees that
	  replacing $\tlambda$ by $\tlambda^\gamma$ for $\gamma=-{\rm deg} (h)>0$ allows us to conclude that the lowest degree component of $\left({\widetilde{\lambda}}^\gamma z^{\gamma_{1}}\cdots z^{\gamma_{d}}\right)^Q\widetilde{\mathcal{P}}(z,\tlambda^{-\gamma})$ is given by $\widetilde{a}_1(z,\tlambda^\gamma)$,
	  where 
	  \begin{equation}\label{g11}
\widetilde{a}_1(z,\tlambda^\gamma)=\left(\widetilde{\lambda}^\gamma z^{\gamma_{1}}_{1}\cdots z^{\gamma_{d}}_{d}\right)^Q\prod_{n\in W}(h(\mu_n\odot z) -{\widetilde{\lambda}}^{-\gamma})= (z_1^{\gamma_1-\gamma_1^\prime}\cdots z_d^{\gamma_d-\gamma_d'})^Q\widetilde{a}(z,\tlambda^\gamma).
	  \end{equation}
	  We denote by $\tf_l^1(z,\tlambda^{\gamma})$  the lowest degree components of $\tf_l(z,\tlambda^{\gamma})$, $l=1,2,\ldots,m$. From (\ref{contradicteq}) it must be that
	  \begin{equation}
\prod_{l=1}^m\tf_l^1 (z,\tlambda^{\gamma})=\widetilde{a}_1(z,\tlambda^\gamma)
	  \end{equation}
	  and hence
	    \begin{equation}\label{g12}
	  \prod_{l=1}^m\tf_l^1 (z,\tlambda)=\widetilde{a}_1(z,\tlambda).
	  \end{equation}
	  Given $l\in\{1,\ldots,m\}$, $\tf_l^1 (z,\tlambda)$ is a polynomial in $z_1^{q_1},z_2^{q_2}, \ldots, z_d^{q_d}$. Thus, there exists 
	  $f_l^1 (z,\tlambda)$ such that 
	  \begin{equation}\label{g13}
  \tf_l^1 (z,\tlambda)=  f_l^1 (z_1^{q_1},\ldots,z_d^{q_d},\tlambda). 
	  \end{equation}
	  By \eqref{g11}, \eqref{g12} and \eqref{g13}, 
	   we reach, recalling the definition of ${a}(z,\tlambda)$ in (\ref{tiladfn}),
	     \begin{equation}
	   \prod_{l=1}^mf_l^1 (z,\tlambda)=\left( z_1^{\frac{\gamma_1-\gamma_1'}{q_1}} z_2^{\frac{\gamma_2-\gamma_2'}{q_2}} \cdots z_d ^{\frac{\gamma_d-\gamma_d'}{q_d}}  \right)^Q{a}(z,\tlambda).
	   \end{equation}
	   By Lemma \ref{airred}, ${a}(z,\tlambda)$ is irreducible, so there exists $j\in\{1,2,\ldots,m\}$ such that 
	   $f_j^1 (z,\tlambda)$ has a factor ${a}(z,\tlambda)$. 
	  	We  conclude that the highest power of $\tlambda$ in $\tf_j(z,\widetilde{\lambda})$ (hence in $f_j (z,\tlambda)$) is at least $Q$.  
	  	Since $m>1$ and, by Lemma~\ref{Lemma:irredmeets} and Claim ~\ref{cl:mainproof}, $\tf_l (z,\tlambda)$ , $l=1,2,\ldots, m$,  must depend on $\tlambda$ we reach a contradiction since the highest power of $\tlambda$ on the left-hand side of \eqref{contradicteq} is equal to $Q$. 
	    \end{proof}
	    \section{Floquet Theory for Long-Range Operators} \label{sec:floquet}

Let us summarize some of the important points about Floquet theory for operators with long-range interactions.
This is well-known, especially in the continuum case; see the survey \cite{Kuchment2016BAMS} and references therein. We are unaware of a precise reference in the discrete setting for long-range operators, so we included the details for the reader's convenience.

Let us assume that $A:\ell^2(\bbZ^d) \to \ell^2(\bbZ^d)$ is bounded. Writing $A_{n,m} = \langle \delta_n,A\delta_m\rangle$ for the matrix elements, we further assume that $A$ is translation-invariant in the sense that 
\[
A_{n+\ell,m+\ell} = A_{n,m} \ \forall n,m,\ell \in \bbZ^d,
\] 
and that $A$ satisfies the decay estimate
\[|A_{n,m}| \leq C e^{-\nu|n-m|}\]
for constants $C,{\nu}>0$. By translation-invariance, $A$ is fully encoded by $\{a_n := A_{n,0}\}_{n \in \bbZ^d}$ via
\[ [A\psi]_n = \sum_{m \in \Z^d} a_{n-m} \psi_m.\]

We denote the Fourier transform on $\ell^2(\bbZ^d)$ by $\scrF:u \mapsto \widehat{u}$, where
\[\widehat{u}(x) =\sum_{n \in \bbZ^d} e^{-2\pi i \langle n,x\rangle} u_n, \]
for $u \in \ell^1(\bbZ^d)$ and then extended to $\ell^2$ by Plancherel.

By the assumptions on $A$, the \emph{symbol} $\widehat{a}$ is analytic, real-analytic whenever $a_n = \overline{a_{-n}}$, and a trigonometric polynomial whenever $a$ is finitely supported. For example, when $A = -\Delta$ denotes the Laplacian on $\bbZ^d$, 
\[\widehat{a}(x) = -2\sum_{j=1}^d\cos(2\pi x_j).\]

Recall that $V:\bbZ^d \to \bbC$ is $q$-periodic and $\Gamma= \{q \odot m : m \in \bbZ^d\}$ denotes the period lattice.
We  define the dual lattice $\Gamma^* = \{ (m_1/q_1,\ldots,m_d/q_d) : m_j \in \bbZ\}$ and
\[W^* := \Gamma^* \cap [0,1)^d = \set{0,\frac{1}{q_1}, \ldots,\frac{q_1-1}{q_1}} \times \cdots \times \set{0,\frac{1}{q_d}, \ldots,\frac{q_d-1}{q_d}}. \]
  The discrete Fourier transform of a $q$-periodic function $g:\bbZ^d \to \bbC$ is defined by
\begin{equation}\label{eq:discretefourier} \widehat{g}_\ell = \frac{1}{\sqrt{Q}} \sum_{n \in W} e^{-2\pi i \langle n, \ell \rangle}g_n, \quad \ell \in W^*.
\end{equation}
Of course, this also makes sense for $\ell \in \Gamma^*$ and satisfies $\widehat{g}_{\ell+n}= \widehat{g}_{\ell}$   for any $\ell \in W^*$ and any $n\in \bbZ^d$. One can check the inversion formula
\begin{align} \label{eq:floq:inversion}
\frac{1}{\sqrt{Q}}
\sum_{\ell \in W^*} e^{2\pi i \langle \ell,n \rangle} \widehat{g}_{\ell}
& = g_n, \ \forall n \in \bbZ^d,
\end{align}
which holds for any $q$-periodic $g$.

Let $\bbT^d = \bbR^d/\bbZ^d$ denote the torus.

\begin{prop}
For any $f \in L^2(\bbT^d)$,
\[[\scrF A \scrF^* f](x)= \widehat{a}(x) f(x)\]
and
\[[\scrF V \scrF^* f](x)
= \frac{1}{\sqrt{Q}} \sum_{\ell \in {W}^*}  \widehat{V}_\ell f(x-\ell).\]
\end{prop}

\begin{proof}
These follow from  direct calculations using the definitions of and assumptions on $A$ and $V$ and the inversion formula \eqref{eq:floq:inversion}. 
\end{proof}

Let us now define $\bbT^d_* = \bbR^d/\Gamma^*$,
\[\scrH_q =  \int_{\bbT^d_*}^\oplus \bbC^W \, \frac{dx}{|\bbT^d_*|} = L^2\left(\bbT^d_*,\bbC^W; \frac{dx}{|\bbT^d_*|}\right) \] 
and $\scrF_q : \ell^2(\bbZ^d) \to \scrH_q$ by $u \mapsto \widehat{u}$ where
\[\widehat{u}_j(x) = \sum_{n \in \bbZ^d} e^{-2 \pi i \langle n \odot q,x\rangle}u_{j+n\odot q}, \ x \in \bbT^d_*, \ j \in W.\]
As usual, this is initially defined for (say) $\ell^1$ vectors, but  has a unique extension to a unitary operator on $\ell^2$ via Plancherel.

\begin{prop}
The operator $\scrF_q$ is unitary. If $V$ is $q$-periodic, then 
\[
\scrF_q H \scrF_q^* = \int^\oplus_{\bbT^d_*} \widetilde{H}(x) \, \frac{dx}{|\bbT^d_*|},
\]
where $ \widetilde{H}(x)$ denotes the restriction of $H$ to $W$ with boundary conditions 
\begin{equation}\label{g1}
u_{n+m\odot q}  = e^{2\pi i \langle m\odot q,x\rangle } u_n, \  n,m\in\Z^d.
\end{equation}
\end{prop}
\begin{proof} Unitarity of $\scrF_q$ follows from Parseval's formula. The form of $\scrF_q H \scrF_q^*$ follows from a direct calculation.\end{proof}

Given $x \in \bbR^d$, let $\scrF^{x}$ be the Floquet-Bloch transform defined on $\bbC^W$ as follows: for any vector on $W$, $\{u(n)\}_{n\in W}$, we set
\[ [\scrF^{x} u]_l 
= \frac{1}{\sqrt{Q}} \sum_{n \in W} e^{-2\pi i  \sum_{j=1}^d  (\frac{l_j}{q_j}+x_j) n_j }u_n, \quad l\in W .\]

Therefore,
\[ [(\scrF^{x} )^*u]_l 
= \frac{1}{\sqrt{Q}} \sum_{n \in W} e^{2\pi i  \sum_{j=1}^d  (\frac{n_j}{q_j}+x_j) l_j }u_n, \quad l\in W .\]
Let $z_j=e^{2\pi i x_j}$, $j=1,2,\cdots,d$ and
define the Laurent series $p(z)$ by
\begin{equation}
p(e^{2\pi i x_1},e^{2\pi i x_2},\ldots, e^{2\pi i x_d}) ={\widehat{a}}(x_1,x_2,\ldots,x_d).
\end{equation} 
Using multi-index notation, we may rewrite this as
\begin{equation}\label{pzequation}
 p(z) = \widehat{a}(x) = \sum_{n \in \bbZ^d} e^{-2\pi i \langle n,x\rangle} a_n = \sum_{n \in \bbZ^d}  a_nz_1^{-n_1} z_2^{-n_2}\cdots z_d^{-n_d} = \sum_{n \in \bbZ^d} a_n z^{-n}.
\end{equation}
\begin{prop} \label{prop:floquetTransf}
Assume $V$ is $q$-periodic.
Then 
$\widetilde{H}(x)$ given by \eqref{g1}  is unitarily  equivalent to 	
$
D^z+B_V,
$
where $z_j = e^{2\pi i x_j}$, $D^z$ is a diagonal matrix with entries
\begin{equation}\label{D}
D^z(n,n^\prime) = p(\mu_{n}\odot z)\delta_{n,n^{\prime}},
\end{equation}
$\mu_n$ is the vector from \eqref{action}, and $B=B_V$ has entries related to the discrete Fourier transform of $V$ via
$$ B(n,n')=\widehat V\left(\frac{n_1-n'_1}{q_1},\ldots,\frac{n_{d}-n'_d}{q_d}\right).$$ 
\end{prop}

\begin{remark}
In particular, $D^z$ depends on $A$ and is independent of $V$, while $B_V$ depends only on $V$ with no dependence on $A$.
\end{remark}

\begin{proof}[Proof of Proposition~\ref{prop:floquetTransf}]
By a direct calculation, we see that $\scrF^{x}$ is unitary, so	it suffices to prove that $D^z+B_V=  (\scrF^{x})\widetilde{H}(x) (\scrF^{x} )^*$.  
  Let $\widetilde{H}_0(x)$ be $\widetilde{H}(x)$ with the potential $V$ set to zero. We are going to show  $(\scrF^{x})\widetilde{H}_0(x) (\scrF^{x})^* =D^z$ and  $(\scrF^{x})V(\scrF^{x})^* =B$ separately.
To prove that $(\scrF^{x})\widetilde{H}_0(x) (\scrF^{x})^* =D^z$,
it suffices to show that for any $u=\{u_n\}_{n\in W}$,
\begin{equation*}
(\scrF^{x} )^*D^z u=\widetilde{H}_0(x) (\scrF^{x} )^*u .
\end{equation*}
It is worth mentioning that $(\scrF^{x} )^*u$ satisfies  \eqref{g1}  so that $\widetilde{H}_0(x)( \scrF^{x})^* u$ is well defined.
With the given definitions, for any $m\in W$, 
\begin{align}\nonumber
(\widetilde{H}_0(x)(\scrF^{x})^* u)_m
&=  \sum_{l\in\Z^d} a_{m-l }[(\scrF^{x} )^*u]_l \\
\nonumber
&= \frac{1}{\sqrt{Q}} \sum_{l\in\Z^d} a_{m-l}\sum_{n \in W} e^{2\pi i  \sum_{j=1}^d (\frac{n_j}{q_j}+x_j) l_j }u_n\\
\nonumber
&= \frac{1}{\sqrt{Q}} \sum_{l\in\Z^d} a_{l}\sum_{n \in W} e^{2\pi i  \sum_{j=1}^d  (\frac{n_j}{q_j} +x_j) (m_j-l_j) }u_n\\
&= \frac{1}{\sqrt{Q}} 
\sum_{n \in W} e^{2\pi i  \sum_{j=1}^d  (\frac{n_j}{q_j} +x_j) m_j } {\widehat{a}}\left(\frac{n_1}{q_1}+x_1, \ldots, \frac{n_d}{q_d}+x_d\right)  u_n.\label{g2}
\end{align}
Putting together \eqref{eq:characterActionDef} and \eqref{g2},
\begin{equation}\label{g4}
{\widehat{a}}\left(\frac{n_1}{q_1}+x_1, \ldots, \frac{n_d}{q_d}+x_d\right) = p(\mu_{n} \odot z) = D^z(n,n).
\end{equation}
 Similarly,
\begin{align}
((\scrF^{x} )^*D^z u)_m &= \frac{1}{\sqrt{Q}} \sum_{n \in W} e^{2\pi i  \sum_{j=1}^d (\frac{n_j}{q_j} +x_j) m_j }u_n D^z(n,n).\label{g3}
\end{align}
 By \eqref{g2}, \eqref{g3} and \eqref{g4}, we finish the proof of $(\scrF^{x})\widetilde{H}_0(x) (\scrF^{x} )^{*}=D^z$.

The proof of $(\scrF^{x})V(\scrF^{x} )^{*}=B$ is similar.
\end{proof}

\section{Proof of Theorem~\ref{t:blochIrr} and Examples} \label{sec:examples}

\begin{proof}[Proof of Theorem~\ref{t:blochIrr}]
The Bloch variety precisely consists of those $(k,\lambda)$ such that there is a nontrivial solution of $Hu=\lambda u$ satisfying the boundary conditions as in \eqref{eq:blochvarDef}.
Thus, with $D$ and $B$ as in Proposition~\ref{prop:floquetTransf}, the Bloch variety is the zero set of the Laurent polynomial $\mathcal{P}(z,\lambda)$ defined by \eqref{Pdef} where
\[\widetilde{\mathcal{P}}(z,\lambda) = \det(D^z+B-\lambda I).\]
After using the standard permutation expansion for this determinant, we see that $\widetilde{\mathcal{P}}$ is of the form \eqref{mathcalP} (with $p$ given via \eqref{pzequation}).
By a brief calculation, one can check that $\widetilde{\mathcal{P}}$ satisfies \eqref{eq:mathcalPCharInv}. Namely, if $S_m$ denotes the shift $e_n \mapsto e_{n+m}$ with addition computed modulo $\Gamma$, one can check that
\begin{align*}
\widetilde{\mathcal{P}}(\mu_mz,\lambda)
& = \det(D^{\mu_m z}+B-\lambda) \\
& = \det(S_m^* D^z S_m + B-\lambda).
\end{align*}
Since  $S_m^* B S_m = B$, \eqref{eq:mathcalPCharInv} follows.
Thus, the result follows from Theorem~\ref{mainthm}.
\end{proof}

Let us conclude by discussing a few examples of how to obtain the generator $p(z)$ for which Theorem~\ref{mainthm} is applicable. In particular, the examples below show that the framework of the present paper allows one to consider different discrete geometries. We start with the most basic example of the Laplacian on $\bbZ^d$, where \begin{equation*}
    [A \psi]_{n} 
= - \sum_{\|m-n\|_1=1}\psi_{m}.
\end{equation*}
In this case, it readily follows from (\ref{pzequation}) that
\begin{equation} \label{eq:sqsymbol}
p(z)=-\left(z_1+\frac{1}{z_1}+z_2+\frac{1}{z_2}+\cdots+z_d+\frac{1}{z_d}\right).
\end{equation}
\begin{proof}[Proof of Corollary~\ref{coro:square}]
From \eqref{eq:sqsymbol}, we see that the minimal degree component of $p$ is precisely
\begin{equation*}
h(z)=-\left(\frac{1}{z_1}+\frac{1}{z_2}+\cdots+\frac{1}{z_d}\right).
\end{equation*}
Here assumptions~\ref{assump1} and~\ref{assump2} are fulfilled with $\mathrm{deg}(h)=-1$.
\end{proof}

We then proceed to the description of a couple of two dimensional examples. The triangular lattice is given by specifying the vertex set 
\[\mathcal{V} 
= \{ nb_1+nb_2 : n,m \in \Z \}, \quad b_1 = \begin{bmatrix} 1 \\ 0 \end{bmatrix}, \ b_2 = \frac{1}{2}\begin{bmatrix} 1 \\ \sqrt{3} \end{bmatrix}\]
with edges given by $u\sim v \iff \|u-v\|_2 = 1$. Applying the shear transformation $b_1\mapsto b_1$, $b_2 \mapsto [0,1]^\top$, one can view this graph as having vertices in $\Z^2$ and
\[ u \sim v \iff u-v \in \{\pm e_1, \pm e_2, \pm(e_1-e_2)\}. \]

In particular, the nearest-neighbor Laplacian on the triangular lattice is equivalent to the operator $A_{\rm tri}:\ell^2(\bbZ^2) \to \ell^2(\bbZ^2)$ such that
\[[A_{\rm tri} \psi]_{n_1,n_2} 
= -\psi_{n_1-1,n_2}-\psi_{n_1+1,n_2}
-\psi_{n_1,n_2-1}-\psi_{n_1,n_2+1}
-\psi_{n_1-1,n_2+1}-\psi_{n_1+1,n_2-1}. \]

 Making use of \eqref{pzequation} one finds that 
\begin{equation} \label{eq:trisymb} 
p_{\rm tri}(z)=-\left(z_1+\frac{1}{z_1}+z_2+\frac{1}{z_2}+\frac{z_1}{z_2}+\frac{z_2}{z_1}\right).
\end{equation}
\begin{proof}[Proof of Corollary~\ref{coro:tri}]
 From \eqref{eq:trisymb}, we see that
 \[h_{\rm tri}(z)=-\frac{1}{z_1}-\frac{1}{z_2},\]
 from which it is trivial to check Assumptions~\ref{assump1} and \ref{assump2}.
 \end{proof}

Finally, in the Extended Harper Model
\begin{align*}
[A_{\rm EHM} \psi]_{n_1,n_2} 
=& -\psi_{n_1-1,n_2}-\psi_{n_1+1,n_2}
-\psi_{n_1,n_2-1}-\psi_{n_1,n_2+1}\\
&-\psi_{n_1-1,n_2+1}-\psi_{n_1+1,n_2-1}-\psi_{n_1-1,n_2-1}-\psi_{n_1+1,n_2+1}.
\end{align*}
Equation (\ref{pzequation}) now implies that
\[p_{\rm EHM}(z)=-\left(z_1+\frac{1}{z_1}+z_2+\frac{1}{z_2}+\frac{z_1}{z_2}+\frac{z_2}{z_1}+z_1z_2+\frac{1}{z_1z_2}\right)\]
The lowest degree component is $\frac{1}{z_1z_2}$. The proof of Corollary~\ref{coro:ehm} follows in just the same way as before; notice that we need the periods to be coprime to ensure that Assumption~\ref{assump2} is met.

\begin{remark} \label{rem:zdgraph}
To conclude, let us say a few words about when our results can be applied in the general setting of periodic operators on periodic graphs. That is, given a periodic graph and a periodic operator thereupon, when can one apply a suitable change of coordinates as in the case of the triangular lattice to reduce an operator on $\bbZ^d$? In short, this can be done to periodic operators on any $\bbZ^d$-periodic graph with transitive vertex action. Let us describe this in a little more detail.

Suppose $\mathcal{G}$ is a locally finite graph with vertices $\mathcal{V}$ on which $\bbZ^d$ acts freely by graph automorphisms. Denote the action of $n \in \bbZ^d$ on the vertex $u \in \mathcal{V}$ by $n + u$. This gives a unitary representation of $\bbZ^d$ viz.\ $[T^n\psi](u) = \psi(n+ u)$, $\psi \in \ell^2(\mathcal{V})$.
One may consider operators of the form $H = A+V$ acting in the natural Hilbert space $\ell^2(\mathcal{V})$, where
\begin{enumerate}
\item[(1)] $A$ commutes with the $\bbZ^d$-action (that is, $AT^n=T^nA$ for all $n \in \bbZ^d$)
\item[(2)] $V$ is diagonal, that is, $[V\psi](u) = V(u)\psi(u)$ for a suitable function $V:\mathcal{V} \to \bbC$.
\item[(3)] $V$ is invariant under the action of a full-rank subgroup of $\bbZ^d$, that is, there is a subgroup $F \leq \bbZ^d$ of rank $d$ such that $VT^m=T^mV$ for all $m \in F$.
\end{enumerate}
In item (3), notice that one can always take $F$ to be of the form $q_1\bbZ \oplus \cdots \oplus q_d \bbZ$ for some $q \in \bbN^d$, in which case we say $V$ is $q$-periodic as before.
If in addition, $\bbZ^d$ acts transitively on $\mathcal{V}$, then choosing arbitrarily some $u_0 \in \mathcal{V}$, one has a one-to-one correspondence $\bbZ^d \to \mathcal{V}$ via $n \mapsto n+ u_0$. Of course, this induces a unitary operator $Q:\ell^2(\bbZ^d) \to \ell^2(\mathcal{V})$ in an obvious manner. In this case, it is clear that $Q^{*}HQ$ is a periodic operator of the form studied in the present paper.
\end{remark}

\bibliographystyle{abbrv}
\bibliography{FLMbib}

\end{document}